\documentclass[12pt,reqno]{article}
\usepackage[english]{babel}
\usepackage{amsmath,dsfont, amsthm}
\usepackage{amsfonts}
\usepackage{graphicx}
\usepackage{enumerate}
\usepackage[usenames,dvipsnames]{color}
\usepackage{subfigure}
\usepackage[right,pagewise,displaymath, mathlines]{lineno}
\usepackage{epstopdf}
\usepackage{color}
\usepackage{multirow}
\usepackage[table,xcdraw]{xcolor}
\usepackage[title]{appendix}

\usepackage{authblk}
\usepackage[round]{natbib}
\usepackage{float}

\usepackage{enumitem}

\usepackage{tikz}
\usepackage{schemabloc}

\usepackage[bottom]{footmisc}

\usepackage[colorlinks = true,urlcolor = blue,citecolor = blue,breaklinks]{hyperref}

\usepackage{geometry}
\numberwithin{equation}{section}
\numberwithin{figure}{section}
\numberwithin{table}{section}

\newtheorem{theorem}{Theorem}[section]
\newtheorem{corollary}{Corollary}[section]
\newtheorem{lemma}{Lemma}[section]

\theoremstyle{definition}

\newtheorem{remark}{Remark}[section]

\numberwithin{equation}{section}

\definecolor{darkread}{rgb}{0.7, 0, 0}
\definecolor{darkbrown}{rgb}{0.55, 0.2, 0.15}
\definecolor{darkblue}{rgb}{0.1,0.1,0.6}
\definecolor{darkgreen}{rgb}{0.1,0.5,0.2}

\newcommand{\dd}{\mathrm{d}}

\newcommand{\R}{\mathbb{R}}

\newcommand{\ES}{\mathrm{ES}}
\newcommand{\rem}{\mathrm{rem}}

\usepackage{setspace}
\usepackage[bottom]{footmisc}

\title{\Large\bf Assessing the difference between integrated quantiles and integrated cumulative distribution functions\thefootnote\relax\footnotetext{We are indebted to Ruodu Wang for reading several drafts of this paper, and for a host of inspiring comments and suggestions. We are also grateful to three anonymous reviewers and the editor in charge of our manuscript for constructive criticism, queries, and suggestions. Our research has been supported by the NSERC Alliance--MITACS Accelerate grant entitled ``New Order of Risk Management: Theory and Applications in the Era of Systemic Risk'' from the Natural Sciences and Engineering Research Council (NSERC) of Canada, and the national research organization Mathematics of Information Technology and Complex Systems (MITACS) of Canada, as well as by the individual NSERC Discovery Grants of Y.~Wei (RGPIN-2023-04674) and R.~Zitikis (RGPIN-2022-04426).}}

\author[,1]{Yunran Wei \thanks{Corresponding author; e-mail: \href{mailto:	Yunran.Wei@carleton.ca}{	Yunran.Wei@carleton.ca}}}
\author[,2,3]{Ri\v{c}ardas Zitikis \thanks{e-mail: \href{mailto:rzitikis@uwo.ca}{rzitikis@uwo.ca}}}

\affil[1]{\normalsize School of Mathematics and Statistics, Carleton University, \break Ottawa, Ontario~K1S~5B6, Canada}
\affil[2]{\normalsize School of Mathematical and Statistical Sciences, Western University, \break London, Ontario~N6A~5B7, Canada}
\affil[3]{\normalsize Risk and Insurance Studies Centre, York University, \break Toronto, Ontario M3J 1P3, Canada}

\date{}

\begin{document}


\maketitle

\vspace{-8mm}

\noindent
\textbf{Abstract.}
This paper offers a mathematical invention that shows how to convert integrated quantiles, which often appear in risk measures, into integrated cumulative distribution functions, which are technically more tractable from various perspectives. The invention helps to avoid a number of technical assumptions that have been traditionally imposed when working with quantities containing quantiles. In particular it helps to completely avoid the requirement of the existence of a probability density function. The developed results explain and illustrate the invention, whose byproducts include the assessment of model uncertainty and misspecification, and the derivation of statistical inference results.


\newpage

\section{Introduction}

A number of problems in mathematical finance and insurance rely on risk measures, a large number of which are quantile based. Many of them are weighted integrals, or other functionals, of the underlying quantile functions, also known as Values-at-Risk (VaR). A few illustrative examples are:
\begin{itemize}
\item
Distortion (spectral) risk measures \citep[e.g.,][Section~8.2.1]{MFE15}.
\item
Expected Shortfall (ES), also known as the Tail Conditional Expectation, in addition to a number of other names.
\item
Range-Value-at-Risk (RVaR) \citep{CDS10}, which is the average of quantiles that bridges the ES and the VaR.
\item
Gini Shortfall (GS) \citep{FWZ2017}.
\item
Inter-ES \citep{BFWW22}, which is a variability measure defined as the difference of the ES's at different levels.
\end{itemize}
Standard references for the mathematical theory of risk measures are \cite{pr2007}, \cite{R13}, \cite{MFE15}, and \cite{FS16}.

Formally, let $X$ be a real-valued random variable, whose cumulative distribution function (cdf) we denote by $F$. Suppose for the sake of illustration that we are interested in developing a large-sample non-parametric statistical inference for the integral
\begin{equation}
\int_p^1 F^{-1}(u) \dd u
\label{int-0}
\end{equation}
for some probability level $p\in (0,1)$, where
\[
F^{-1}(u)=\inf\{ x\in \mathbb{R}: F(x)\geq u \}
\]
is the $u^{\textrm{th}}$ quantile of the cdf $F$, that is, the VaR at the level $u$.

Before we proceed further, we need to introduce additional notation. Namely, let $\mathcal{F}_1^+$ denote the set of all cdf's $F$ for which integral~\eqref{int-0} is finite. This is equivalent to saying that $\mathcal{F}_1^+$ is the set of all cdf's $F$ such that the random variables $X\sim F$ satisfy $\mathbb{E}(X^+)<\infty $, where $X^+ =\max\{X,0\}$. Obviously, $\mathcal{F}_1^+ \supset \mathcal{F}_1$, where  $\mathcal{F}_1$ is the set of all cdf's $F$ for which integral~\eqref{int-0} is finite when $p=0$, that is, consists of all those cdf's $F$ such that the random variables $X\sim F$ have finite first moments $\mathbb{E}(X) $. (Recall that $\mathbb{E}(X) $ is finite if and only if $\mathbb{E}(X^+)<\infty $ and $\mathbb{E}(X^-)<\infty $, where $X^- =\max\{-X,0\}$.) The class of all cdf's is denoted by $\mathcal{F}$.

We shall now introduce yet another cdf, which in Section~\ref{maths} below will be a generic cdf denoted by $G$, but presently, to initiate the reader's intuition and to also connect the topic of the present paper to what is already known in the literature, we choose to work with the empirical cdf $F_n$ defined by
\begin{equation}\label{emp-cdf}
F_n(x)={1\over n} \sum_{i=1}^n \mathds{1}\{X_i\le x\},
\end{equation}
where, for illustrative purposes, we assume that the random variables $X_1,\dots , X_n$ are independent copies of $X$. Hence, non-parametric statistical inference for $F^{-1}(u)$ is based on the empirical $u^{\textrm{th}}$ quantile
\[
F_n^{-1}(u)=\inf\{ x\in \mathbb{R}: F_n(x)\geq u \}.
\]
Establishing the limiting distribution for the appropriately normalized difference $F^{-1}(u)-F_n^{-1}(u)$ is challenging because it is not the average of transformed random variables $X_1,\dots , X_n$, although under some assumptions (e.g., absolute continuity of the cdf $F$ plus other minor assumptions on the probability density function, pdf), the difference $F^{-1}(u)-F_n^{-1}(u)$ is, asymptotically when $n\to \infty $, such an average \citep{B1966}. For more details on the topic, we refer to, e.g.,  \citet[][Sections~2.5 and~2.8.3]{S1980}. Hence, it is tempting to conclude that under the same assumptions, the integral
\begin{equation}\label{diff-int}
\int_p^1 \big( F^{-1}(u)-F_n^{-1}(u) \big) \dd u
\end{equation}
is also, asymptotically when $n\to \infty $, the average of certain transformations of $X_1,\dots , X_n$. This is indeed true but, very interestingly, such an asymptotic representation of integral~\eqref{diff-int} holds under much weaker assumptions than those required for $F^{-1}(u)-F_n^{-1}(u)$.

To see why integration improves the situation, we set $p=0$ and write the equations
\begin{align}
\int_0^1 \big( F^{-1}(u)-F_n^{-1}(u) \big) \dd u
&=\int_{-\infty}^{\infty } \big( F_n(x)-F(x) \big) \dd x
\notag
\\
&= {1\over n}\sum_{i=1}^n \int_{-\infty}^{\infty } \big( \mathds{1}\{X_i\le x\}-F(x) \big) \dd x ,
\label{eq-0}
\end{align}
with an illuminating proof of the first equation given in Lemma~\ref{lemma-1}. Hence, except for the inevitable requirement that $F$ has a finite first moment, that is, $F\in \mathcal{F}_1$, no other assumption is required for equation~\eqref{eq-0} to hold,  and this has inspired our current considerations, developed in full generality in next Section~\ref{maths} and illustrated throughout the rest of this paper with examples spanning areas well beyond statistical inference.

Indeed, the first equation of~\eqref{eq-0} naturally leads us to the topic of the next section where, for two arbitrary cdf's $F$ and $G$, we assess the magnitude of the ``gap''
\begin{equation}\label{gap}
\Gamma_p(F,G):=\int_p^1 \big( F^{-1}(u)-G^{-1}(u) \big) \dd u
-
\int_{F^{-1}(p) }^{\infty } \big( G(x)-F(x) \big) \dd x
\end{equation}
between the two integrals on the right-hand side of equation~\eqref{gap}. Obviously, our earlier illustration concerns with the special case $G=F_n$, but our main results hold for generic cdf's $G$ and are therefore formulated and discussed in this way in Section~\ref{maths}. Apart from the traditional in the  mathematical sciences strive to obtain as general results as possible, the generality of our arguments in the next section is welcome from several perspectives:
\begin{enumerate}
  \item Having the main results only in the case $G=F_n$ would potentially mislead the reader into thinking that the empirical cdf $F_n$ is necessary for our arguments, which is not the case as only the very basic properties of cdf's are actually needed.
  \item The usefulness of our arguments is much wider than the mere case of $F_n$  and includes prominent scenarios such as model uncertainty, misspecified distributions (think of the mixture of the underlying cdf $F$ and some other cdf $H$), and various parametric, non-parametric and other estimators of $F$, depending on sampling designs, which could, and in practice are, rather complex. We shall elaborate on these topics in concluding  Section~\ref{conclusion}, when all the required for such a discussion results have been established. 
\end{enumerate}

Foundational results for $\Gamma_p(F,G)$ in the case of generic pairs $(F,G)$ of cdf's, and also for other related to $\Gamma_p(F,G)$ quantities, are formulated and discussed in Section~\ref{maths}. Section~\ref{stats} contains several corollaries in the special case $G=F_n$ that illustrate how statistical inference for  integrated quantile~\eqref{int-0} and its various functionals can almost effortlessly be derived from the results of Section~\ref{maths}.
As a further illustration of the power of our general results of Section~\ref{maths}, in Section~\ref{crm} we shall discuss coherent distortion (spectral) risk measures, including the ES. Section~\ref{conclusion} concludes the paper with additional notes and afterthoughts. Proofs and other technicalities are in Appendix~\ref{proofs}.

\section{Foundational results}
\label{maths}

For the sake of symmetry and thus added mathematical beauty, in the following theorem we consider the difference
\begin{equation}\label{gamma-2d}
 \Delta_{p,z} (F,G)  = \int_p^1 ( F^{-1}(u) -G^{-1}(u) )\dd u  -  \int_z^\infty (G(x)-F(x)) \dd x
\end{equation}
between the two integrals for arbitrary $p\in (0,1)$ and $z\in \R$. Of course, setting $z=F^{-1}(p)$ brings us back to the original task of assessing the magnitude of $\Gamma_p(F,G)$ because
\begin{equation}\label{gamma-0}
\Gamma_p(F,G)= \Delta_{p,F^{-1}(p)} (F,G).
\end{equation}

\begin{theorem}\label{th:new}
Let $p\in (0,1)$ and $z\in \R$. Then
\begin{align}
 \big (F(z) -p\big )\big ( F^{-1}(p)-z\big )\le  \Delta_{p,z} (F,G)
 \le  \big (G(z)-p\big ) \big (z-G^{-1}(p)\big ).
\label{eq:new}
 \end{align}
\end{theorem}

The ``difference'' functional $\Delta_{p,z}: \mathcal{F}_1^+ \times \mathcal{F}_1^+ \to \mathbb{R}$ is antisymmetric, that is, $\Delta_{p,z} (F,G) =-\Delta_{p,z} (G,F) $ for all $F,G \in \mathcal{F}_1^+$. Furthermore, given $F,G \in \mathcal{F}_1^+$, the real-valued function $(p,z)\mapsto \Delta_{p,z}(F,G)$ is well defined on the strip $(0,1)\times \R$, is always finite, and vanishes when the pair $(p,z)$ is equal to $ (0,-\infty)$ or  $(1,\infty )$.

Note that the product on the left-hand side of bound~\eqref{eq:new} is always non-positive, whereas the product on the right-hand side of bound~\eqref{eq:new} is always non-negative, because  for every cdf $H$, and thus for $F$ and $G$ in particular, the bound $H(z)\ge p$ holds if and only if $z\ge H^{-1}(p)$.

It is also important to note that when $z=F^{-1}(p)$, the difference  $\Delta_{p,z}(F,G)$, which is equal to $\Gamma_p(F,G)$,  is non-negative for every $p\in (0,1)$, and when $z=G^{-1}(p)$, the difference is non-positive for every $p\in (0,1)$. We shall see the value of these observations later in this section.

To illustrate the ``difference'' function $(p,z)\mapsto \Delta_{p,z}(F,G)$, let $F\sim \textrm{Lomax}(\alpha_1,1)$ and  $G\sim \textrm{Lomax}(\alpha_2,1)$ be two Lomax cdf's with the shape parameters $\alpha_1>0$ and $\alpha_2>0$, and the same scale parameter $\lambda =1$. We have
\begin{multline}\label{foot-00}
\Delta_{p,z}(F,G)=
\frac{\alpha_1}{\alpha_1-1}(1-p)^{1-1/\alpha_1}
-\frac{\alpha_2}{\alpha_2-1}(1-p)^{1-1/\alpha_2}
\\
- \frac{1}{\alpha_1-1}(1+z)^{1-\alpha_1}
+\frac{1}{\alpha_2-1}(1+z)^{1-\alpha_2},
\end{multline}
which is depicted in Figure~\ref{fig:foot-00}.
\begin{figure}[h!]
    \centering
    \resizebox{80mm}{75mm}{\includegraphics{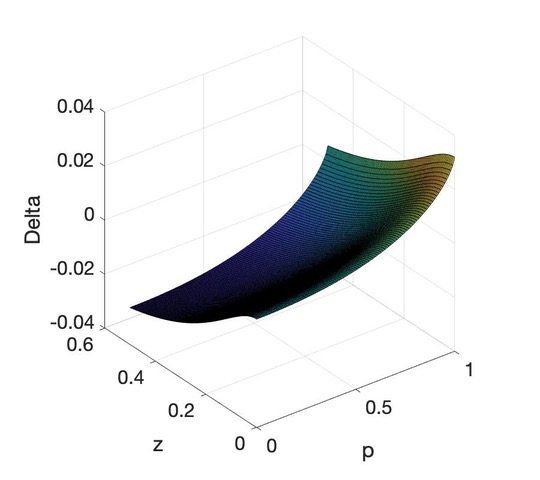}} \quad
    \resizebox{80mm}{75mm}{\includegraphics{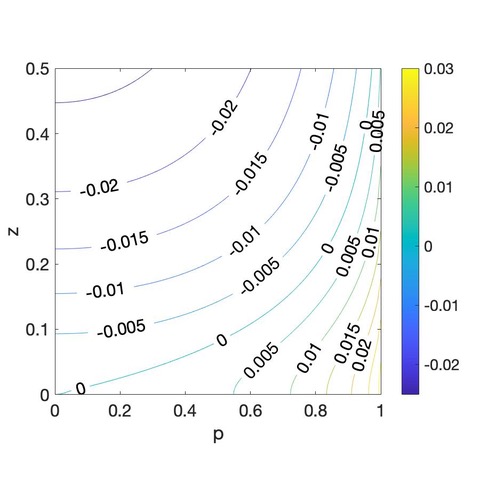}}
    \\
    \resizebox{80mm}{75mm}{\includegraphics{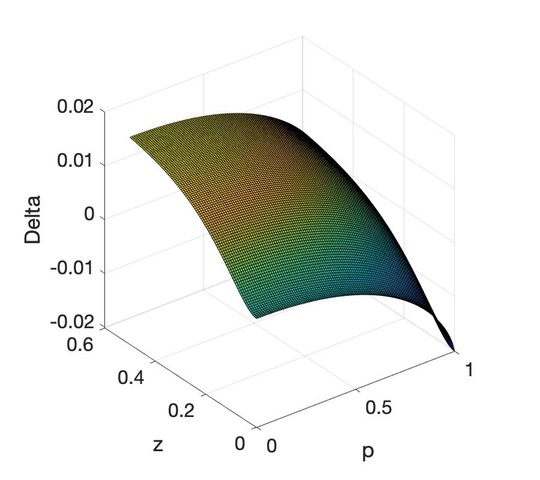}} \quad
    \resizebox{80mm}{75mm}{\includegraphics{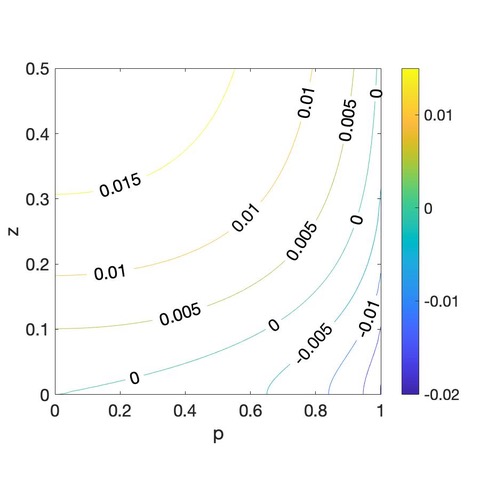}}
    \caption{The ``difference'' function $(p,z)\mapsto \Delta_{p,z}(F,G)$ (left-hand panels) and its contour plots (right-hand panels) when $F\sim\textrm{Lomax}(10,1)$ and $G\sim\textrm{Lomax}(\alpha_2,1)$ with $\alpha_2=8$ (top panels) and $\alpha_2=12$ (bottom panels).}
    \label{fig:foot-00}
\end{figure}
Recall that the Lomax distribution $\textrm{Lomax}(\alpha,\lambda )$, whose cdf is $x \mapsto 1-(1+x/\lambda)^{-\alpha }$, has a finite first moment when the shape parameter $\alpha >1$, and it has a finite variance when $\alpha > 2$. In the two panels of Figure~\ref{fig:foot-00} we have depicted some of these cases.

The following corollary to Theorem~\ref{th:new} plays a fundamental role in the development of statistical inference for integral~\eqref{int-0} in the following section, where we shall set $G=F_n$. Throughout the rest of the current section, however, we keep on working with generic cdf's $F$ and $G$, which belong to either $\mathcal{F}_1^+$ or $\mathcal{F}$, depending on the results considered.

\begin{corollary}\label{cor-00}
For any $p\in (0,1)$, we have
\begin{equation}\label{lembound-00}
0\le \Gamma_p(F,G)\le \big (G(x_p)-p\big ) \big( F^{-1}(p)-G^{-1}(p) \big) ,
\end{equation}
where $x_p:=F^{-1}(p)$. If the cdf $F$ is continuous at the $p^{\textrm{th}}$ quantile  $x_p$, then
\begin{equation}\label{lembound-0}
0\le \Gamma_p(F,G)\le  \big( G(x_p)-F(x_p) \big) \big( F^{-1}(p)-G^{-1}(p) \big).
\end{equation}
\end{corollary}

Hence, when $F,G \in \mathcal{F}_1^+$, the ``gap'' function $p\mapsto \Gamma_p(F,G)$ is well defined on the unit interval $(0,1)$, is always finite, non-negative, and vanishes at $p=0$ and $p=1$. From the definition of $\Gamma_p(F,G)$ we notice the lack of symmetry between the cdf's $F$ and $G$, and this is actually beneficial when developing statistical inference, as we shall see in the following sections.

\begin{remark}
As a little curiosity that immediately follows from Corollary~\ref{cor-00}, we note that if $F^{-1}(p)=G^{-1}(p)$, then $\Gamma_p(F,G)=-\Gamma_p(G,F)$, and since both $\Gamma_p(F,G)$ and $\Gamma_p(G,F)$ are non-negative, they are equal to $0$. Of course, the latter statement also immediately follows from the right-most bound of~\eqref{lembound-0}.
\end{remark}

To illustrate $\Gamma_p(F,G)$, let $F\sim \textrm{Lomax}(\alpha_1,1)$ and  $G\sim \textrm{Lomax}(\alpha_2,1)$. We have
\begin{equation}\label{foot-1}
\Gamma_p(F,G)=
(1-p)^{(\alpha_1-1)/\alpha_1}
+\frac{\alpha_2}{1-\alpha_2}(1-p)^{(\alpha_2-1)/\alpha_2}
- \frac{1}{1-\alpha_2}(1-p)^{(\alpha_2-1)/\alpha_1},
\end{equation}
which, as a function of $p$, is depicted in Figure~\ref{fig:bridge}
\begin{figure}[h!]
    \centering
    \resizebox{80mm}{65mm}{\includegraphics{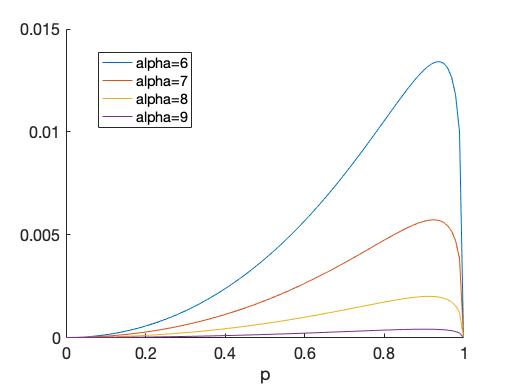}} \quad
    \resizebox{80mm}{65mm}{\includegraphics{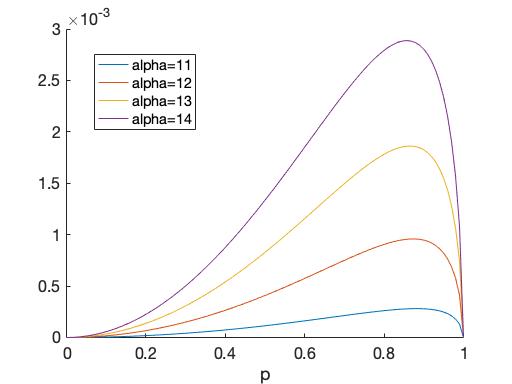}}
    \caption{The ``gap'' function $p\mapsto \Gamma_p(F,G)$ when $F\sim\textrm{Lomax}(10,1)$ and $G\sim\textrm{Lomax}(\alpha_2,1)$ with $\alpha_2= 6, 7, 8$, $9$ (left-hand panel) and $\alpha_2= 11, 12, 13$, $14$ (right-hand panel).}
    \label{fig:bridge}
\end{figure}
for various shape parameter values. Note that $\Gamma_p(F,G)=0$ for all $p\in (0,1)$ when $F=G$, which says that the horizontal axis depicts the function $p\mapsto \Gamma_p(F,G)$ when $G\sim\textrm{Lomax}(\alpha_2,1)$ with the same shape parameter $\alpha_2= 10$ as that of $F$.

Although the definition of $\Gamma_p(F,G)$ requires a finite (upper) first moment, the left-hand side of equation~\eqref{foot-1} is well-defined even when the Lomax shape parameters are below $1$, meaning that the cdf's $F$ and $G$ do not have finite first moments. We shall next explain this phenomenon in an illuminating way via an extension of the functional $\Gamma_p: \mathcal{F}_1^+ \times \mathcal{F}_1^+ \to \mathbb{R}$ to the largest-possible domain $\mathcal{F} \times \mathcal{F}$, where $\mathcal{F}$ is the set of all cdf's, irrespective of whether they have finite moments or not. This makes the contents of the following theorem.

\begin{theorem}\label{theorem-1}
Let $p\in (0,1)$, and let $\Gamma_p^*: \mathcal{F} \times \mathcal{F} \to \mathbb{R}$ be the functional defined by
\begin{equation}\label{extend-1a}
\Gamma_p^*(F,G):= \int_{F^{-1}(p) }^{G^{-1}(p)} \big( p-G(x) \big) \dd x .
\end{equation}
We have the following statements:
\begin{enumerate}[label={\rm\arabic*)}]
\item
If $ F,G\in \mathcal{F}$, then
\begin{equation}\label{extend-1cc}
0\le \Gamma_p^*(F,G)\le  \big( G(x_p)-p \big) \big( F^{-1}(p)-G^{-1}(p) \big),
\end{equation}
and if, additionally, the cdf $F$ is continuous at the $p^{\textrm{th}}$ quantile  $x_p:=F^{-1}(p)$, then
\begin{equation}\label{extend-1c}
0\le \Gamma_p^*(F,G)\le  \big( G(x_p)-F(x_p) \big) \big( F^{-1}(p)-G^{-1}(p) \big).
\end{equation}
\item
If $ F,G\in \mathcal{F}_1^+$, then
\begin{equation}\label{extend-1b}
\Gamma_p^*(F,G)= \Gamma_p(F,G).
\end{equation}
\end{enumerate}
\end{theorem}

Hence, $p\mapsto \Gamma_p^*(F,G)$ is an extended ``gap'' function defined on the unit interval $(0,1)$, always non-negative, taking finite values whenever $p\in (0,1)$, and finite or infinite at the end-points $p=0$ and $p=1$ of its domain of definition.
\begin{figure}[h!]
    \centering
    \resizebox{80mm}{65mm}{\includegraphics{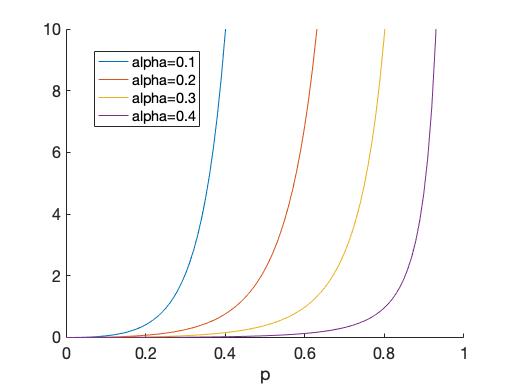}}
    \resizebox{80mm}{65mm}{\includegraphics{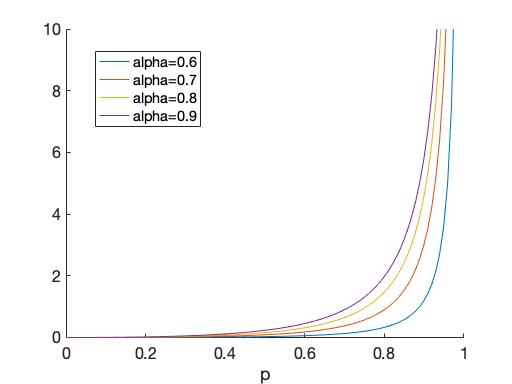}}
    \caption{The extended ``gap'' function $p\mapsto \Gamma_p^*(F,G)$ when $F\sim\textrm{Lomax}(0.5,1)$ and $G\sim\textrm{Lomax}(\alpha_2,1)$ with $\alpha_2= 0.1, 0.2, 0.3$, $0.4$ (left-hand panel) and $\alpha_2= 0.6, 0.7, 0.8$, $0.9$ (right-hand panel).}
    \label{fig:bridge2}
\end{figure}
To illustrate, the left-hand panel of  Figure~\ref{fig:bridge2}
depicts
\begin{equation}\label{foot-2}
\Gamma_p^*(F,G)=
(1-p)^{(\alpha_1-1)/\alpha_1}
+\frac{\alpha_2}{1-\alpha_2}(1-p)^{(\alpha_2-1)/\alpha_2}
- \frac{1}{1-\alpha_2}(1-p)^{(\alpha_2-1)/\alpha_1}
\end{equation}
as a function of $p$ when $F\sim\textrm{Lomax}(0.5,1)$ and $G\sim\textrm{Lomax}(\alpha_2,1)$ with various parameter $\alpha_2$ values strictly below $1$.
Hence, the two cdf's $F$ and $G$ are ultra-heavily tailed, because they do not possess finite first moments. Technically, therefore, $\Gamma_p(F,G)$ does not exist, although $\Gamma_p^*(F,G)$ does exist. Since $\Gamma_p^*(F,G)=0$ for all $p\in (0,1)$ when $F=G$, the horizontal axis depicts $p\mapsto \Gamma_p^*(F,G)$ in the case  $G\sim\textrm{Lomax}(\alpha_2,1)$ with the same shape parameter $\alpha_2= 0.5$ as that of $F$.

In summary, therefore, it is not $\Gamma_p(F,G)$ as the whole that imposes moment-type conditions on the underlying cdf's $F$ and $G$ but the very basic objects that do so, which are the two integrals whose difference makes up the definition of $\Gamma_p(F,G)$ and which require such conditions. This insight, by the way, may potentially lead to the development of ES-type risk measures in situations when the first moments of underlying risks are  infinite, and such situations do exist \citep[e.g.,][]{NECD2006}.

\section{An excursion into statistical inference}
\label{stats}

There are extensive studies devoted to statistical estimation of the ES and other risk measures. Common approaches include parametric methods such as the maximum likelihood and the method of trimmed moments \citep[e.g.,][]{BJZ2009}, semiparametric methods such as those based on Extreme Value Theory \citep[e.g.,][]{EKM97,NRZ10,GGPQ2022}, and non-parametric methods \citep[e.g.,][]{JZ2003,BJPZ2008,C08,PQWY2012}. In terms of condition minimality, the study of \citet{BJPZ2008} is perhaps the closest one to our current study, although this comment applies to only the ES. Indeed, in the case of the Tail Capital Allocation (TCA), which generalizes the ES, a non-parametric methodology has been developed by \citet{GSZ2022a,GSZ2022b}. It should be noted at this point that although the complexities of the latter two studies are unavoidable when dealing with the TCA, they almost vanish in the case of the ES, as we shall soon see.

Hence, let $X_1,\dots , X_n$ be any random variables, and let $F_n$ be their empirical cdf. By setting $G=F_n$ in Corollary~\ref{cor-00}, we readily arrive at the following corollary that plays a pivotal role throughout the rest of this section.

\begin{corollary}\label{corollary-0}
Let $p\in (0,1)$ and $F\in \mathcal{F}_1^+$. Furthermore, let $F_n$ be the empirical cdf based on $X_1,\dots , X_n$. Then
\begin{equation}\label{b-50}
0\le \Gamma_p(F,F_n)\le \big( F_n(x_p)-p \big) \big( F^{-1}(p)-F_n^{-1}(p) \big) .
\end{equation}
If, additionally, the cdf $F$ is continuous at $x_p:=F^{-1}(p)$, then
\begin{equation}\label{b-5}
0\le \Gamma_p(F,F_n)\le \big( F_n(x_p)-F(x_p) \big) \big( F^{-1}(p)-F_n^{-1}(p) \big) .
\end{equation}
\end{corollary}

To appreciate Corollary~\ref{corollary-0}, assume  that $X_1,\dots , X_n$ are independent copies of $X$ whose cdf is $F\in \mathcal{F}_1^+$.
If the quantile function $F^{-1}$ is continuous at the point $p$, meaning that the cdf $F$ is strictly increasing at the point $x_p$, then $F^{-1}(p)-F_n^{-1}(p)=o_{\mathbb{P}}(1)$, and since $F_n(x_p)-p=O_{\mathbb{P}}(1)$, we therefore have
\begin{equation}\label{b-6}
\Gamma_p(F,F_n)=o_{\mathbb{P}}(1)
\end{equation}
when $n\to \infty $. If, on the other hand, $F$ is continuous at the point $x_p$, then $p=F(x_p)$ and so, by the classical law of large numbers, $F_n(x_p)\to_{\mathbb{P}}p=F(x_p)$. Since $F^{-1}(p)\in \mathbb{R}$ and thus $F_n^{-1}(p) =O_{\mathbb{P}}(1)$, we therefore again have statement~\eqref{b-6}. Since any cdf $F$ is either  strictly increasing or continuous, or both, at the point $x_p$, we therefore have the asymptotic representation
\[
\int_p^1 \big( F^{-1}(u)-F_n^{-1}(u) \big) \dd u
= {1\over n}\sum_{i=1}^n Y_{i,p}  +o_{\mathbb{P}}(1)
\]
for every cdf $F\in \mathcal{F}_1^+$, where
\begin{equation}\label{yi}
Y_{i,p}=\int_{F^{-1}(p) }^{\infty } \big( \mathds{1}\{X_i\le x\}-F(x) \big) \dd x .
\end{equation}
The iid random variables $Y_{1,p}, \dots, Y_{n,p}$ have finite first moments because $F\in \mathcal{F}_1^+$. Since their means are zero, by the law of large numbers we have $n^{-1}\sum_{i=1}^n Y_{i,p}  =o_{\mathbb{P}}(1)$ and thus, in turn, we arrive at the following consistency result.

\begin{corollary}\label{cor-2}
Let $p\in (0,1)$ and $F\in \mathcal{F}_1^+$. Furthermore, let $F_n$ be the empirical cdf based on iid~random variables $X_1,\dots , X_n \sim F$.
Then, when $n\to \infty $,
\begin{equation}\label{b-7}
\int_p^1 F_n^{-1}(u)\dd u \to_{\mathbb{P}}\int_p^1 F^{-1}(u) \dd u.
\end{equation}
\end{corollary}

Of course, this corollary can be established in a myriad of other ways and under the same conditions \citep[e.g.,][]{BFWW22}, but the way we have used to prove it here provides an important illustration of how Corollary~\ref{corollary-0} actually works.

\begin{remark}
We have seen that statement~\eqref{b-5} implies  statement~\eqref{b-7} for every $p\in (0,1) $ such that $F^{-1}(p) - F_n^{-1}(p) = o_{\mathbb{P}}(1)$.  The latter statement holds for almost every $p\in (0,1)$, which is a well-known property of empirical quantiles \citep[e.g.,][p.~10]{SW1986}. Therefore, statement~\eqref{b-7} holds for almost every $p\in (0,1)$. Since $p\mapsto \int_p^1 F^{-1}(u) \dd u $ and $p\mapsto \int_p^1 F^{-1}_n(u) \dd u$ are both concave and continuous functions, and since  statement~\eqref{b-7} holds on a dense subset of $(0,1)$, we conclude with the help of   \citet[][Theorem~10.8]{R70} that statement~\eqref{b-7} holds also on the entire $(0,1)$.
\end{remark}

The following CLT-type result is much more useful from the statistical inference point of view than the previous LLN-type result. When reading the following corollary, note the absence of any condition that would involve a pdf of $F$, which is not needed, and is not therefore required to even exist. This is useful and sometimes even crucial because in a number of real-life situations, good cdf models incorporate discrete components, due to the presence of, e.g., many identical values such as claim amounts transformed by  insurance deductibles, policy limits, etc. \citep[e.g.,][and references therein]{br23}.

\begin{corollary}\label{cor-3}
Let $p\in (0,1)$ and $F\in \mathcal{F}_1^+$. Furthermore, let $F_n$ be the empirical cdf based on iid~random variables $X_1,\dots , X_n \sim F$.
If the cdf $F$ is continuous and strictly increasing at $x_p:=F^{-1}(p)$, and if the variance $\sigma_{F,p}^2$ of the random variable
$\int_{F^{-1}(p) }^{\infty } \big( \mathds{1}\{X\le x\}-F(x) \big) \dd x$
is finite, then, when $n\to \infty $, we have the asymptotic normality result
\begin{equation}\label{b-8}
\sqrt{n}\left( \int_p^1 F_n^{-1}(u)\dd u -\int_p^1 F^{-1}(u) \dd u \right)
\to_{d} \mathcal{N}\big( 0,\sigma_{F,p}^2\big ).
\end{equation}
\end{corollary}

Hence, we now require $F$ to be continuous and strictly increasing at $x_p$. To see why we need the latter (strict monotonicity) condition, note that the empirical quantile $F_n^{-1}(p)$ is equal in distribution to $F^{-1}(E_n^{-1}(p))$, where $E_n^{-1}(p)$ is the empirical quantile based on independent and uniformly on the interval $[0,1]$ distributed random variables $U_1,\dots , U_n$, which may be defined on a different probability space if the original one is not rich enough to support such uniform random variables.  (When $F$ is continuous, such uniform random variables always exist in the original space.) It is well known \citep[e.g.,][p.~10]{SW1986} that $E_n^{-1}(p)$ converges in probability to $p$, and so for $F^{-1}(E_n^{-1}(p))$ to converge in probability to $F^{-1}(p)$, we need continuity of $F^{-1}$ at the point $p$, which is equivalent to the assumption that $F$ is strictly increasing at $x_p$. This, by the way, helps us to understand, and appreciate, why \citet{BFWW22}, who assume the existence of a pdf of $F$, require the pdf to be strictly positive in their CLT-type results, as this requirement implies that the cdf $F$ is strictly increasing.

To discuss the variance $\sigma_{F,p}^2$, we need additional notation. Namely, let $\mathcal{F}_2^+$ denote the set of all cdf's $F$ such that any random variable $X\sim F$ satisfies $\mathbb{E}((X^+)^2)<\infty $. Obviously, $\mathcal{F}_2^+ $ consists of all cdf's $F$ for which $\int_p^1 \big(F^{-1}(u)\big)^2 \dd u <\infty $  for every $p\in (0,1)$. As we shall show in Lemma~\ref{lemma-3} in Appendix~\ref{proofs}, the variance $\sigma_{F,p}^2$ is finite whenever $F\in \mathcal{F}_2^+$. Furthermore, we shall also show in the same lemma that when $F\in \mathcal{F}_2^+$, the variance $\sigma_{F,p}^2$ can be expressed as
\begin{equation}
\sigma_{F,p}^2=\int_{F^{-1}(p) }^{\infty } \int_{F^{-1}(p) }^{\infty }
\big( F(x \wedge y)-F(x)F(y) \big) \dd x \dd y,
\label{int-01}
\end{equation}
where $x \wedge y$ denotes the minimum of $x$ and $y$.

To see how Corollary~\ref{cor-3} almost effortlessly follows from Corollary~\ref{corollary-0}, we first rewrite the definition of $\Gamma_p(F,F_n)$ as follows:
\begin{equation}\label{b-20}
\sqrt{n}\int_p^1 \big( F^{-1}(u)-F_n^{-1}(u) \big) \dd u
= {1\over \sqrt{n}}\sum_{i=1}^n Y_{i,p}  +\sqrt{n} \Gamma_p(F,F_n) ,
\end{equation}
where $Y_{1,p},\dots , Y_{n,p}$ are the random variables defined by equation~\eqref{yi}. Obviously, under the conditions of Corollary~\ref{cor-3}, we have
\[
{1\over \sqrt{n}}\sum_{i=1}^n Y_{i,p}  \to_{d} \mathcal{N}\big( 0,\sigma_{F,p}^2\big ),
\]
and so Corollary~\ref{cor-3} follows provided that
\begin{equation}\label{b-9}
\sqrt{n} \Gamma_p(F,F_n)=o_{\mathbb{P}}(1).
\end{equation}
Bound~\eqref{b-5} plays a pivotal role in establishing statement~\eqref{b-9}, as we shall now demonstrate: First, the classical CLT for iid~Bernoulli random variables implies
$\sqrt{n}  \big( F(x_p)-F_n(x_p) \big)=O_{\mathbb{P}}(1)$,
whereas the assumption that the cdf $F$ is strictly increasing at the point $x_p=F^{-1}(p)$ implies
$F_n^{-1}(p)\to_{\mathbb{P}}F^{-1}(p) $ (it is helpful to now recall the discussion in the paragraph that immediately follows Corollary~\ref{cor-3}).
Hence, statement~\eqref{b-9} holds, and so does Corollary~\ref{cor-3}. In summary, we almost effortlessly established the asymptotic normality of the appropriately normalized integral $\int_p^1 F_n^{-1}(u)\dd u $ under minimal conditions on the cdf $F$.

\section{Coherent distortion risk measures and beyond}
\label{crm}

We can equally successfully and almost effortlessly tackle more complicated integrals,  such as
\[
\rho(F):=\int_0^1 \ES_p(F)\mu(\dd p) ,
\]
where $\mu $ is a measure determined by the context of a specific application, or a theory, and
\[
\ES_p(F)={1\over 1-p}\int_p^1 F^{-1}(u) \dd u
\]
is the Expected Shortfall (ES), whose pivotal role in finance and insurance has been amply discussed, with the first-of-its-kind axiomatic foundation provided by \citet{WZ2021}.

All coherent distortion risk measures can be expressed as $\rho(F)$ \citep[][Proportion~8.18]{MFE15}, and the class of these risk measures coincides with the class of all comonotonic-additive coherent risk measures \citep{K01}. Note also that the point measure $\mu(\{p\})=1$ gives $\rho(F)=\ES_p(F)$, which up to the constant $1/(1-p)$ is equal to integral~\eqref{int-0}.  In fact, $\mu$ can be any signed measure as long as $\rho(F)$ is finite \citep[e.g.,][]{WWW2020}, because in what follows we shall only need the linearity property of the integral with respect to the integrand (e.g., with respect to the quantile function) and not its positivity.

Establishing consistency and asymptotic normality of $\rho(F)$ reduces to establishing the corresponding properties of linear combinations of integrals of the types that appear in the above considerations. Indeed, with the empirical ES defined by
\[
\ES_{p,n}(F)={1\over 1-p}\int_p^1 F_n^{-1}(u) \dd u,
\]
we have
\begin{align}
\int_0^1 \Big( \ES_p(F)-\ES_{p,n}(F) \Big) \mu(\dd p)
&= \int_0^1 \bigg( {1\over n}\sum_{i=1}^n Y_{i,p}  +\Gamma_p(F,F_n) \bigg) {1\over 1-p}\mu(\dd p)
\notag
\\
&= {1\over n}\sum_{i=1}^n \int_0^1     {Y_{i,p}\over 1-p}\mu(\dd p)
+\int_0^1  {\Gamma_p(F,F_n)\over 1-p}\mu(\dd p),
\label{es-rep}
\end{align}
where the random variables $Y_{1,p}, \dots , Y_{n,p}$ are defined by equation~\eqref{yi}. Clearly, under appropriate conditions on the cdf $F$ and measure $\mu $, the random variables
\[
Z_{i,p}:=\int_0^1 {Y_{i,p}\over 1-p}\mu(\dd p), \quad 1\le i \le n,
\]
are iid, centered at $0$, and have finite second moments, thus satisfying the CLT. To verify that the right-most integral in equation~\eqref{es-rep} converges in probability to $0$, we first bound it:
\begin{equation}
\int_0^1  {\Gamma_p(F,F_n)\over 1-p}|\mu |(\dd p)
\le \int_0^1  {\big( F^{-1}(p)-F_n^{-1}(p) \big) \big( F_n(x_p)-F(x_p) \big)\over 1-p}|\mu |(\dd p),
\label{es-bound}
\end{equation}
where $|\mu |$ is the variation of the (possibly signed) measure $\mu $. (Note that the variation and the measure itself are different only in the case of signed measures.) In summary, under the simple random sampling design, from equation~\eqref{es-rep} we immediately deduce the following CLT result
\begin{equation}
\sqrt{n}\int_0^1 \Big( \ES_p(F)-\ES_{p,n}(F) \Big) \mu(\dd p)
\to_{d} \mathcal{N}\big( 0,\sigma_{F,\mu }^2\big ),
\label{es-clt}
\end{equation}
where the asymptotic variance $\sigma_{F,\mu }^2$ is given by the formula
\[
\sigma_{F,\mu }^2=\int_0^1\int_0^1 {\mathbb{E}\big(Y_{1,p}Y_{1,q}\big)\over (1-p)(1-q)}\mu(\dd p)\mu(\dd q)
\]
with
\[
\mathbb{E}\big(Y_{1,p}Y_{1,q}\big)
= \int_{F^{-1}(p) }^{\infty }\int_{F^{-1}(q) }^{\infty }
\big( F(x \wedge y)-F(x)F(y) \big)  \dd x \dd y.
\]

The use the second half of Corollary~\ref{corollary-0} to establish bound~\eqref{es-bound} may at first glance give the impression that $F$ has to be continuous at $x_p$ for every $p\in (0,1)$, but this is true only if we ignore the role of $\mu $. To illustrate how important it is to take the measure $\mu $ into account, we start with the simplest example $\mu(\{p^*\})=1$ with any fixed $p^*\in (0,1)$, in which case the validity of bound~\eqref{es-bound} follows if we assume that $F$ is continuous at $x_{p^*}$ for the given $p^*$. In the case of discrete mixtures of ES's, we would have $\mu(\cup_{i}\{p_i\})=\sum_{i} \mu(\{p_i\})=1$ and thus bound~\eqref{es-bound} would hold whenever the cdf $F$ is continuous at $x_p$ for every $p\in \cup_{i}\{p_i\}$. For continuous measures, the matter is simpler because the cdf $F$ and the quantile function $F^{-1}$ can be discontinuous only on at most countable number of points.

Hence, coming back to bound~\eqref{es-bound} and assuming its validity (i.e., assuming  appropriate conditions on $F$ and $\mu $), we can employ weighted LLN- and  CLT-type results for general quantile and empirical processes \citep[e.g.,][]{SW1986} in order to show that the right-hand side of bound~\eqref{es-bound} converges in probability to $0$ when $n\to \infty $. These are standard technicalities, whose choices are contingent on the available information (or lack of it) about the measure $\mu $ and the cdf $F$, and, also very importantly, on how the two interact.

To illustrate, consider a simple (in the context of the present paper) but very important case of the Expected Shortfall at any given probability level $p^* \in (0,1)$, which we briefly mentioned above but will now tackle with rigour and in full detail. Hence, with $\mu $ being the probability measure induced by the degenerate at the point $p^* \in (0,1)$ random variable, we have the equation
\[
\sqrt{n}\int_0^1 \Big( \ES_p(F)-\ES_{p,n}(F) \Big) \mu(\dd p)
=\sqrt{n}\Big( \ES_{p^*}(F)-\ES_{p^*,n}(F) \Big)
\]
and, in view of statement~\eqref{es-clt} and the surrounding it discussion, we have the asymptotic normality result
\[
\sqrt{n}\Big( \ES_{p^*}(F)-\ES_{p^*,n}(F) \Big)
\to_{d} \mathcal{N}\big( 0,\sigma_{F}^2\big )
\]
with the asymptotic variance
\begin{equation}
\sigma_{F}^2={1\over (1-p^*)^2}\int_{F^{-1}(p^*) }^{\infty }\int_{F^{-1}(p^*) }^{\infty }
\big( F(x \wedge y)-F(x)F(y) \big)  \dd x \dd y,
\label{es-var}
\end{equation}
provided that the following two conditions hold: first, $\mathbb{E}((X^+)^2)<\infty $, and second, the cdf $F$ is continuous and strictly increasing at the quantile $F^{-1}(p^*)$.  In the current context, these are truly minimal conditions.

It now becomes instructive to recall the work of \citet[][Section~5]{BFWW22} whose expression
\begin{equation}
\sigma_{F}^2={1\over (1-p^*)^2}\int_{p^*}^{1}\int_{p^*}^{1}
{ s \wedge t-st \over f(F^{-1}(s))f(F^{-1}(t)) } \dd s \dd t
\label{es-var-pdf}
\end{equation}
of the asymptotic variance $\sigma_{F}^2$ is of course equivalent to that given by equation~\eqref{es-var}, provided that the cdf $F$ has a density $f$, which we do not require due to our technique of proof. The reason \citet{BFWW22} need absolute continuity of the cdf $F$ is that their proof, which is quite different from ours, relies on reducing the asymptotic behaviour of
$\sqrt{n}( F^{-1}(u)-F_n^{-1}(u)) $ to that of a weighted Brownian bridge \citep{B1966}, thus inevitably requiring the existence of $f$.

Finally, note the following alternative way of writing equation~\eqref{es-var-pdf}:
\[
\sigma_{F}^2={1\over (1-p^*)^2}\int_{p^*}^{1}\int_{p^*}^{1}
( s \wedge t-st )\dd F^{-1}(s) \dd F^{-1}(t).
\]
It does not rely on the existence of $f$, and this alternative expression of the asymptotic variance $\sigma_{F}^2$ in the form of a Lebesgue-Stieltjes integral serves a strong indication that absolute continuity of the cdf $F$ is not needed, and we have indeed established this fact in the present paper.

\section{Concluding notes}
\label{conclusion}

The main goal of this paper has been to show that under very mild assumptions,  integrated quantiles can be converted into integrated cdf's with an error term for which theoretically and practically useful bounds have been derived and illustrated. Apart from being an interesting mathematical result, one of the biggest benefits of such a conversion is statistical, which could be at the population level (e.g., assessing model uncertainty or misspecification) or at the data level (e.g., assessing the performance of various estimators).

Consider first a problem at the population level, inspired by \citet{CDS10}. Specifically, the results that we have derived in the previous sections can be used to assess model uncertainty of the tail behaviour of risks by considering, e.g., a set $\mathcal{H}$ of misspecified cdf's such that each $F\in \mathcal{H}$ is only a small perturbation away from the true cdf. Let $F_0$ denote the (unknown) true cdf of the population whose ES at a level $p\in (0,1)$ we wish to assess. The expert's subject-matter knowledge may suggest some cdf $F\in \mathcal{H}$ as a proxy for $F_0$. Given this information, the resulting $\ES_p(F)$ is known, but what can we say about the ``ideal'' $\ES_p(F_0)$, assuming that the expert believes -- with confidence -- that $F$ is  within a certain distance from $F_0$?

It should be noted at this point that the closeness of $F$ and $F_0$ on their domains of definition $\mathbb{R}$ does not automatically imply the closeness of the corresponding values-of-risk, that is, of the quantiles on their domains of definition $(0,1)$. Hence the challenge, and bound~\eqref{lembound-00} with $F_0$ instead of $G$ gives a helping hand in sorting out the problem:
\begin{equation}\label{gap-i}
\ES_p(F)-\ES_p(F_0)
= {1\over 1-p}\int_{F^{-1}(p) }^{\infty } \big( F_0(x)-F(x) \big) \dd x
+\rem_p(F,F_0),
\end{equation}
where the (non-negative) remainder term $\rem_p(F,F_0)$ satisfies the bound
\begin{equation}\label{lembound-00-i}
\rem_p(F,F_0)\le {1\over 1-p}\big (F_0(F^{-1}(p))-p\big ) \big( F^{-1}(p)-F_0^{-1}(p) \big) .
\end{equation}
The main term on the right-hand side of equation~\eqref{gap-i} is tractable, given the expert's subject-matter knowledge of the quantile $F^{-1}(p)$ and an estimate of the distance between the cdf's $F$ and $F_0$. In view of this knowledge, the right-hand side of bound~\eqref{lembound-00-i} is also tractable, provided that, additionally, we can assess the closeness of the $p^{\textrm{th}}$ quantiles (i.e., values-at-risk) $F^{-1}(p)$ and $F_0^{-1}(p)$. In summary, therefore, to assess the distance between $\ES_p(F)$ and $\ES_p(F_0)$, in addition to what is already known to the expert, we also need to assess the distance between the quantiles  $F^{-1}(p)$ and $F_0^{-1}(p)$. This is a considerably lesser problem than assessing the distance between the two quantile functions on their domains of definition $(0,1)$.

Consider now a basic though quite illuminating ``statistical'' example. Namely, by their very definition (recall equation~\eqref{emp-cdf}), empirical cdf's are sums of random variables, and thus integrals of empirical cdf's are also sums of random variables. This linearity plays a pivotal role when establishing desired statistical inference results for integrated quantiles and thus, in turn, for various risk measures of insurance and finance. Elaborating on this statistical aspect, in the previous sections we have shown the validity of the following results:
\begin{itemize}
\item
If $F\in \mathcal{F}_1^+$, then $\int_p^1 F_n^{-1}(u)\dd u $ is a consistent estimator of $\int_p^1 F^{-1}(u)\dd u $.
\item
If the cdf $F\in \mathcal{F}_2^+$ is continuous and strictly increasing at $x_p$, then $\int_p^1 F_n^{-1}(u)\dd u $ is asymptotically normal.
\end{itemize}
These results have been established under the iid~assumption on $X_1,\dots , X_n$, but this assumption can be relaxed, and thus the two results can be established in various non-iid~scenarios (e.g., under $\alpha $-mixing, etc.), as required by specific applications.

Of course, when working with profit-and-loss (P\&L) distributions, the left-hand version $\int_0^p F^{-1}(u)\dd u $ of integral~\eqref{int-0} is also of interest, and for it, we have the following analogs of the above statements:
\begin{itemize}
\item
If $F\in \mathcal{F}_1^-$, then $\int_0^p  F_n^{-1}(u)\dd u $ is a consistent estimator of $\int_0^p  F^{-1}(u)\dd u $.
\item
If $F\in \mathcal{F}_2^-$ and the cdf $F$  is continuous and strictly increasing at $x_p$, then $\int_0^p  F_n^{-1}(u)\dd u $ is asymptotically normal.
\end{itemize}
The sets $\mathcal{F}_1^-$ and $\mathcal{F}_2^-$ are defined like $\mathcal{F}_1^+$ and $\mathcal{F}_2^+$, respectively, but now using the negative part $X^- =\max\{-X,0\}$ instead of $X^+$.

We can of course equally successfully and almost effortlessly tackle more complicated integrals such as $\int_{\Delta} F_n^{-1}(u)\dd u $, as long as $\Delta $ is the union of some disjoint subintervals of $(0,1)$. Indeed, establishing consistency and asymptotic normality for such integrals reduces to establishing the corresponding properties of linear combinations of integrals of the types that we have extensively  discussed in the current paper, and the results such as those discussed by \citet[][Section~3.3]{S1980} make the task almost effortless.

\appendix
\section{Proofs}
\label{proofs}

Before proving Theorems~\ref{th:new} and \ref{theorem-1}, we shall first establish two auxiliary lemmas: the first one will confirm the validity of the first equation of~\eqref{eq-0}, whereas the second lemma will be used in the proof of Theorem~\ref{th:new}. Note that in the two lemmas, as well as when proving the two theorems and establishing other technical results, we shall avoid the classical formulas of integration-by-parts and change-of-variables. Instead, we shall rely on Fubini's theorem, which is perfectly suited for our purpose, especially in view of the fact that, in general, the cdf's and their quantile functions are not, strictly speaking, the ordinary inverses of each other.

\begin{lemma}\label{lemma-1}
If $F\in \mathcal{F}_1$, that is, if the first moment of $X$ is finite, then
\begin{equation}\label{lemma-eq-1}
\int_0^1 \big( F^{-1}(u)-F_n^{-1}(u) \big) \dd u
=\int_{-\infty}^{\infty } \big( F_n(x)-F(x) \big) \dd x.
\end{equation}
\end{lemma}

\begin{proof}
We begin with the obvious equations
\begin{align}
\int_0^1 \big( F^{-1}(u)-F_n^{-1}(u) \big) \dd u
&=\int_0^1 F^{-1}(u) \dd u
-\int_0^1 F_n^{-1}(u) \dd u
\notag
\\
&= \mathbb{E}(X)-\bar{X},
\label{lemma-eq-1a}
\end{align}
where $\bar{X}$ denotes the sample mean of $X_1,\dots , X_n$.  To show that the right-hand sides of equations~\eqref{lemma-eq-1} and~\eqref{lemma-eq-1a} are equal, we shall employ Fubini's theorem. To avoid notational confusion, we shall use $x_1,\dots , x_n$ instead of $X_1,\dots , X_n$, that is, we shall prove the equation
\begin{equation}\label{lemma-eq-1b}
{1\over n}\sum_{i=1}^n \int_{-\infty}^{\infty } \big( \mathds{1}\{x_i\le x\}-F(x) \big) \dd x
=\mathbb{E}(X)-\bar{x}
\end{equation}
We now write a string of equations:
\begin{align}
\int_{-\infty}^{\infty } &\big( \mathds{1}\{x_i\le x\}-F(x) \big) \dd x
\notag
\\
&=
\int_{-\infty}^{\infty } \mathds{1}\{x_i\le x\}\big( 1-F(x) \big) \dd x
+\int_{-\infty}^{\infty } \mathds{1}\{x_i> x\}\big( -F(x) \big) \dd x
\notag
\\
&= \mathbb{E}\bigg(\int_{-\infty}^{\infty } \mathds{1}\{x_i\le x\}\mathds{1}\{X> x\}  \dd x\bigg)
-\mathbb{E}\bigg(\int_{-\infty}^{\infty } \mathds{1}\{x_i> x\}\mathds{1}\{X\le x\} \dd x \bigg)
\notag
\\
&= \mathbb{E}\big( (X-x_i)_{+}\big)-\mathbb{E}\big( (x_i-X)_{+}\big)
\notag
\\
&= \mathbb{E}\big( (X-x_i)_{+}\big)-\mathbb{E}\big( (X-x_i)_{-}\big)
\notag
\\
&= \mathbb{E}(X)-x_i,
\label{lemma-eq-1c}
\end{align}
where $a_{+}$ and $a_{-}$ denote, respectively, the positive and negative parts of any $a\in \mathbb{R}$ and satisfy the equation $a_{+}-a_{-}=a$. Equation~\eqref{lemma-eq-1c} obviously leads to equation~\eqref{lemma-eq-1b}, thus concluding the proof of Lemma~\ref{lemma-1}.
\end{proof}

\begin{lemma}\label{lemma-2}
If $F\in \mathcal{F}_1^+$, that is, if $\mathbb{E}(X^+)<\infty $, then the equation
\begin{equation}\label{lemma-eq-2}
\int_z^\infty (1-F(x)) \dd x =  \int_\R (x-z)_+ \dd F(x)
\end{equation}
holds for every $z\in \mathbb{R}$.
\end{lemma}

\begin{proof}
Just like in the proof of the previous lemma, we rely on Fubini's theorem and have the equations
\begin{align*}
\int_z^\infty (1-F(x)) \dd x
&= \mathbb{E}\bigg(\int_{-\infty}^{\infty } \mathds{1}\{x>z\}\mathds{1}\{X> x\}  \dd x\bigg)
\notag
\\
&= \mathbb{E}\big( (X-z)_{+}\big)
\notag
\\
&=  \int_\R (x-z)_+ \dd F(x).
\end{align*}
This concludes the proof of Lemma~\ref{lemma-2}.
\end{proof}

\begin{proof}[\bf Proof of Theorem~\ref{th:new}]
For any $H\in \mathcal{F}_1^+$, and thus for $F,G\in \mathcal{F}_1^+$ in particular, let $m_H(z)$ be the function defined by
\begin{align*}
m_H(z)
:=& (1-p)z + \int_z^\infty (1-H(x)) \dd x
\\
=&  (1-p)z + \int_\R (x-z)_+ \dd H(x) ,
\end{align*}
where we used Lemma~\ref{lemma-2}.
Hence,
\begin{align}   \notag
-\int_z^\infty (G(x)-F(x)) \dd x
&=   \int_z^\infty (1-G(x)) \dd x    - \int_z^\infty (1-F(x)) \dd x
\\&= \int_\R (x-z)_+ \dd G(x)  -  \int_\R (x-z)_+ \dd F(x) \notag
\\&= m_G(z)-m_F(z) . \label{eq:new-0}
 \end{align}
Using the ES formula  of \citet[][Theorem~10]{RU02}, we have
\begin{align}
\int_p^1  F^{-1}(u) \dd u
&= \min_{y\in \R} \left\{(1-p)y
+ \int_\R (x-y)_+ \dd F(x)\right\}
\notag
\\
&=\min_{y\in \R} m_F(y) .
\label{eq:new-1}
 \end{align}
Similarly, equation~\eqref{eq:new-1} holds with $F$ replaced by $G$. Putting equations~\eqref{eq:new-0} and \eqref{eq:new-1} together, we arrive at
 \begin{align*}
 \Delta_{p,z} (F,G)  & =\min_{y\in \R}m_F(y)  -  \min_{y\in \R} m_G(y)
 + m_G(z)-m_F(z)
 \\& \le m_G(z) -  \min_{y\in \R} m_G(y).
\end{align*}
Note that the function $m_G(z)$ is convex and its right-hand derivative is $D_{+}m_G(z)=G(z)-p$. Since $G^{-1}(p) \in \arg\min_{y\in \R} m_G(y)$, we therefore have \citep[e.g.,][p.~61]{W1991}
\begin{align*}
m_G(z) -  \min_{y\in \R} m_G(y)
&\le D_{+}m_G(z) ( z-G^{-1}(p))
\\
&= (G(z) -p)( z-G^{-1}(p)).
\end{align*}
This establishes the right-hand bound of~\eqref{eq:new}.
For the left-hand bound, we apply the just established result on $\Delta_{p,z} (G,F)$ and have
\begin{align*}
 - \Delta_{p,z} (F,G)
 &= \Delta_{p,z} (G,F)
 \\
 &\le (F(z) -p)( z-F^{-1}(p)).
\end{align*}
This establishes the left-hand bound of~\eqref{eq:new} and completes the proof of Theorem~\ref{th:new}.
\end{proof}

\begin{proof}[\bf Proof of Theorem~\ref{theorem-1}.]
We start by proving  statement~\eqref{extend-1b}, that is, we first show that  $\Gamma_p(F,G)$ coincides with $\Gamma_p^*(F,G)$ whenever $ F,G\in \mathcal{F}_1^+$. Note that $\Gamma_p(F,G)$ is finite due to $ F,G\in \mathcal{F}_1^+$, and so we write the equation
\[
\int_{F^{-1}(p) }^{\infty } \big( F(x)-G(x) \big) \dd x
= \int_{G^{-1}(p) }^{\infty } \big( 1-G(x) \big) \dd x
-\int_{F^{-1}(p) }^{\infty } \big( 1-F(x) \big) \dd x
+ \int_{F^{-1}(p) }^{G^{-1}(p)} \big( 1-G(x) \big) \dd x .
\]
Fubini's theorem implies
\begin{align}
\int_{F^{-1}(p) }^{\infty } \big( 1-F(x) \big) \dd x
&=\mathbb{E}\bigg( \int_{F^{-1}(p) }^{\infty } \mathds{1}\{X> x\} \dd x \bigg)
\notag
\\
&= \mathbb{E}\Big( \big( X-F^{-1}(p))_{+} \Big )
\notag
\\
&= \int_0^1 \big( F^{-1}(u)-F^{-1}(p) \big)_{+} \dd u
\notag
\\
&= \int_p^1 F^{-1}(u)\dd u -(1-p)F^{-1}(p).
\label{eq-2}
\end{align}
Of course, the same equations hold when $F$ is replaced by $G$. Combining the equations, we obtain
\begin{align}
\Gamma_p(F,G)
&= (1-p)\big( F^{-1}(p)-G^{-1}(p) \big)
+ \int_{F^{-1}(p) }^{G^{-1}(p)} \big( 1-G(x) \big) \dd x
\notag
\\
&= \int_{F^{-1}(p) }^{G^{-1}(p)} \big( p-G(x) \big) \dd x,
\label{b-1}
\end{align}
which establishes statement~\eqref{extend-1b}.

We next prove statement~\eqref{extend-1c}, and thus work with arbitrary cdf's $F,G\in \mathcal{F}$. To show that $\Gamma_p^*(F,G)$ is non-negative, we start with the case $F^{-1}(p) < G^{-1}(p)$. Consequently, the integration variable $x$ in the definition of $\Gamma_p^*(F,G)$ satisfies the inequality $x<G^{-1}(p)$, which is equivalent to $G(x)<p$. This implies $\int_{F^{-1}(p) }^{G^{-1}(p)} \big( p-G(x) \big) \dd x \ge 0$.

When $F^{-1}(p) \ge  G^{-1}(p)$, on the other hand, $\Gamma_p^*(F,G)$ is equal to $\int_{G^{-1}(p) }^{F^{-1}(p)} \big( G(x)-p \big) \dd x $, and since the integration variable $x$ is such that $x\ge G^{-1}(p)$, we have $G(x)\ge p$ and thus $\int_{G^{-1}(p) }^{F^{-1}(p)} \big( G(x)-p \big) \dd x \ge 0$. This concludes the proof that   $\Gamma_p^*(F,G)\ge 0$.

It remains to establish the right-most bounds of statements~\eqref{extend-1cc} and~\eqref{extend-1c}. We  start with the case $F^{-1}(p) < G^{-1}(p)$ and have
\begin{align}
\int_{F^{-1}(p) }^{G^{-1}(p)} \big( p-G(x) \big) \dd x
&\le \int_{F^{-1}(p) }^{G^{-1}(p)} \big( p-G(F^{-1}(p)) \big) \dd x
\notag
\\
&= \big( G^{-1}(p)-F^{-1}(p) \big) \big( p-G(F^{-1}(p)) \big),
\label{b-2}
\end{align}
where the inequality holds because $x\ge F^{-1}(p)$ and thus $G(x)\ge G(F^{-1}(p))$. This establishes statement~\eqref{extend-1cc}. Statement~\eqref{extend-1c} immediately follows from bound~\eqref{b-2} because $F(F^{-1}(p))\ge p$ always holds.

When $F^{-1}(p) \ge G^{-1}(p)$, we have
\begin{align}
\int_{G^{-1}(p) }^{F^{-1}(p)} \big( G(x)-p \big) \dd x
&\le \int_{G^{-1}(p) }^{F^{-1}(p)} \big( G(F^{-1}(p))-p \big) \dd x
\notag
\\
&= \big( F^{-1}(p)-G^{-1}(p) \big) \big( G(F^{-1}(p))-p \big),
\label{b-3}
\end{align}
where the inequality holds because $x\le F^{-1}(p)$ and so $G(x)\le G(F^{-1}(p))$.
This establishes statement~\eqref{extend-1cc}. Statement~\eqref{extend-1c} follows from bound~\eqref{b-3} because the continuity of $F$ at $x_p=F^{-1}(p)$ implies $F(F^{-1}(p))= p$. This finishes the entire proof of  Theorem~\ref{theorem-1}.
\end{proof}

\begin{lemma}\label{lemma-3}
Let $p\in (0,1)$ and $F\in \mathcal{F}_2^+$. Then the variance $\sigma_{F,p}^2$ is finite and can be expressed by formula~\eqref{int-01}.
\end{lemma}

\begin{proof}
To show that the variance is finite, we need to check that
\begin{equation}\label{quant-0}
\mathbb{E}\left( \int_{F^{-1}(p) }^{\infty } \int_{F^{-1}(p) }^{\infty }
\big| \mathds{1}\{X\le x\}-F(x) \big|
\big| \mathds{1}\{X\le y\}-F(y) \big|
\dd x \dd y\right)<\infty .
\end{equation}
This is the same as showing that
\begin{equation}\label{quant-1}
\mathbb{E}\left( \int_{F^{-1}(p) }^{\infty } \int_{F^{-1}(p) }^{\infty }
\big| \mathds{1}\{X> x\}-S(x)\big|
\big| \mathds{1}\{X> y\}-S(y) \big|
\dd x \dd y\right)<\infty ,
\end{equation}
where $S=1-F$ is the survival function.
Since $p\in (0,1)$ and thus $F^{-1}(p)\in \mathbb{R}$, the assumption $F\in \mathcal{F}_2^+$ (which implies $F\in \mathcal{F}_1^+$) together with equation~\eqref{eq-2} imply that the integral $\int_{F^{-1}(p) }^{\infty } \big( 1-F(x) \big) \dd x $ is finite. This reduces checking statement~\eqref{quant-1} to proving
\begin{equation}\label{quant-2}
\mathbb{E}\left( \int_{F^{-1}(p) }^{\infty } \int_{F^{-1}(p) }^{\infty }
\mathds{1}\{X> x\} \mathds{1}\{X> y\}
\dd x \dd y\right)<\infty .
\end{equation}
The double integral is separable, and each of the two integrals is equal to $(X-F^{-1}(p))_{+}$. Consequently, statement~\eqref{quant-2} holds whenever the expectation $\mathbb{E}\left( (X-F^{-1}(p))_{+}^2\right)$ is finite, and the latter holds because $F\in \mathcal{F}_2^+$. Consequently, $\sigma_{F,p}^2<\infty $.

To prove equation~\eqref{int-01}, we start with statement~\eqref{quant-1}, which we have already established. Fubini's theorem can now be applied, thus yielding the equations
\begin{align*}
\mathbb{E}&\left( \int_{F^{-1}(p) }^{\infty } \int_{F^{-1}(p) }^{\infty }
\big( \mathds{1}\{X\le x\}-F(x) \big)
\big( \mathds{1}\{X\le y\}-F(y) \big)
\dd x \dd y\right)
\\
&=\int_{F^{-1}(p) }^{\infty } \int_{F^{-1}(p) }^{\infty }
\mathbb{E}\Big(\big( \mathds{1}\{X\le x\}-F(x) \big)
\big( \mathds{1}\{X\le y\}-F(y) \big)\Big)
\dd x \dd y
\\
&= \int_{F^{-1}(p) }^{\infty } \int_{F^{-1}(p) }^{\infty }
\mathrm{Cov}\big( \mathds{1}\{X\le x\},\mathds{1}\{X\le y\} \big)
\dd x \dd y
\\
&= \int_{F^{-1}(p) }^{\infty } \int_{F^{-1}(p) }^{\infty }
\big( F(x \wedge y)-F(x)F(y) \big) \dd x \dd y.
\end{align*}
This completes the proof of equation~\eqref{int-01}.
\end{proof}

\end{document}